\documentclass[a4,reqno,10pt]{amsart}

\usepackage[foot]{amsaddr}

\usepackage{amsmath,amsthm,amssymb,xcolor,graphicx,soul}
\usepackage{amsfonts}
\usepackage{float}
\usepackage{amscd,verbatim,a4wide}
\usepackage[all]{xy}
\usepackage{braket}
\usepackage[mathscr]{euscript}
\usepackage{dsfont}
\usepackage{mathtools}
\usepackage[customcolors]{hf-tikz}
\usepackage{subcaption}
\usepackage{tikz}
\usepackage{tikz-cd}
\usepackage[most]{tcolorbox}
\usepackage{xargs} 
\usepackage{bm}

\usepackage{hyperref}
\hypersetup{colorlinks=true,	citecolor = blue,linkcolor=black}  

\usepackage[colorinlistoftodos]{todonotes}

\newcommandx{\unsure}[2][1=]{\todo[linecolor=red,backgroundcolor=red!25,bordercolor=red,#1]{#2}}
\newcommandx{\done}[2][1=]{\todo[linecolor=green,backgroundcolor=green!25,bordercolor=green,#1]{#2}}

\def\be{\begin{equation}}
\def\ee{\end{equation}}

\theoremstyle{plain}
\newtheorem{theorem}{Theorem}

\newtheorem{proposition}[theorem]{Proposition}

\theoremstyle{definition}
\newtheorem{definition}[theorem]{Definition}

\theoremstyle{remark}

\newtheorem*{example}{Example}

\numberwithin{equation}{section}
\numberwithin{theorem}{section}
\numberwithin{figure}{section}
\numberwithin{table}{section}

\usepackage{tikz}
\usetikzlibrary{arrows}
\usetikzlibrary{positioning}
\usepackage{amsmath}
\usepackage{amsfonts}
\usepackage{hyperref}
\usepackage{amsthm}
\usepackage{graphicx}

\newcommand{\RR}{\mathbb{R}}

\newcommand{\pr}{\partial}
  
\newcommand{\ZZ}{\mathbb{Z}}

\usepackage[english]{babel}
\usepackage{blindtext}

\begin{document}

\title[Trapped Photons in Schwarzschild-Tangherlini Spacetimes]{Trapped Photons in Schwarzschild-Tangherlini Spacetimes}

\author[M Bugden]{Mark Bugden}
\address[M Bugden]{Mathematical Institute, Faculty of Mathematics and Physics,
	Charles University Prague 186 75, Czech Republic}
\email{bugden@karlin.mff.cuni.cz}

\tikzset{node distance=1in, auto}

\begin{abstract}
We study the set of trapped photons outside the event horizon of Schwarzschild-Tangherlini black holes in dimension $d\geq4$. We prove that the only trapped photons are those with constant Boyer-Lindquist coordinate radius. We then identify the set of trapped photons as a submanifold of the tangent bundle of the spacetime, and prove that this submanifold has topology $T^1 S^{n-1} \times \RR^2$, as suggested in \cite{CJ19}. 

\end{abstract}


\maketitle
\section{Introduction}
\label{sec:Intro}
A curious feature of the Schwarzschild black hole is that, outside the event horizon, there is a radius at which light can orbit. Photons at this radius can move in circular orbits around the black hole, and the collection of such orbits forms a surface called the photon sphere. Although these orbits are unstable under radial perturbations, they are nevertheless interesting from a physical and mathematical point of view since they provide valuable information on geometric and optical properties of the black hole. For static spacetimes in four dimensions, photon spheres feature prominently in uniqueness proofs \cite{C15,CG15,CG16,LY15,Y15,LY16}. In addition, the optical appearance of a star undergoing gravitational collapse \cite{AT68}, the shadow of a black hole \cite{G16}, how the night sky would look to an observer near a black hole \cite{N93}, as well as the analysis of stability for black holes \cite{DR08} are all related to the properties of the photon sphere. For more information on the geometry of photon spheres in four-dimensions, we refer the reader to the article by Claudel, Virbhadra and Ellis \cite{CEV00}. \\

Spherical photon orbits can also exist in more interesting black holes, such as rotating black holes in four \cite{T03} and higher dimensions \cite{B18}, and $\eta$-deformed black holes \cite{DJL17}. For a rotating black hole, the structure of spherical photon orbits changes dramatically. The orbits of photons are no longer contained in a plane, and need not even have a fixed azimuthal direction. For these black holes, the set of trapped photons is no longer a sphere, and is not even a submanifold of spacetime.\footnote{The region accessible to trapping ``pinches" at some points}. It turns out that there is a nicer geometric description of the set of trapped photons utilising the tangent bundle of spacetime. Here, one identifies geodesics in spacetime with points in the tangent bundle. This was done for Schwarzschild and Kerr black holes in 4 dimensions in \cite{CJ19}, where they proved that such a subset is a submanifold of the tangent bundle, and found that the topology of the set of trapped photons was $SO(3) \times \RR^2$.\\

Due to their relevance to string theory and supergravity, higher dimensional black hole solutions have garnered significant attention in recent years. The first string-theoretic calculation of black hole entropy was done for a five-dimensional black hole \cite{SV96}, and five-dimensional gravity is related via the AdS/CFT correspondence to a four-dimensional QFT \cite{M99}. As testing grounds for mathematical relativity, higher-dimensional black holes provide interesting features which are not present in four-dimensional relativity. A striking example is the black-ring solution in five-dimensions \cite{ER02}, which has an event horizon of topology $S^1 \times S^2$. Other novel features are rotating black holes which can rotate in multiple independent planes, or rotating black holes which can have arbitrarily high angular momenta \cite{MP86}. \\

In this paper we study the set of trapped photons - that is, photons which never pass through the horizon or escape to spatial infinity - for the Schwarzschild-Tangherlini black hole spacetimes in dimension $d\geq 4$. The set of spherical photon orbits form a subset of these trapped photons, but it is not \emph{a priori} clear that all trapped photons must have constant coordinate radius. The first result we prove is that the spherical photon orbits are the only trapped photons. Our second result is a characterisation of the topology of the set of trapped photons. We follow \cite{CJ19} in viewing the set of trapped photons as a subset of the tangent bundle of the spacetime, and prove that this subset is a manifold of topology $T^1 S^{n-1} \times \RR^2$, as suggested in \cite{CJ19}.  \\
	
We begin in Section \ref{sec:Notation} by introducing the Schwarzschild-Tangherlini black holes and fixing our notation. In Section \ref{sec:Results} we prove the stated results. Finally, in Section \ref{sec:summary}, we provide a brief summary and discuss our results and future work.  \\


\section{Notation}
\label{sec:Notation}

We begin by introducing the Schwarzschild-Tangherlini black hole in dimension $d = n+1 \geq 4$. These black holes are spherically symmetric, static solutions to the Einstein Field Equations in $d = n+1$ dimensions \cite{T63}. The metric for the Schwarzschild -Tangherlini black hole is
\begin{align}
\label{eq:SchwTangmetric}
ds^2 &= -\left( 1- \frac{\mu}{r^{d-3}} \right) dt^2 + \frac{dr^2}{1- \frac{\mu}{r^{d-3}}} + r^2 d \Omega_{d-2}^2 ,
\end{align}
where $d\Omega_{d-2}^2$ is the metric for the unit $(d-2)$-dimensional sphere. The parameter $\mu$ is related to the mass of the black hole via
\begin{align*}
\mu &= \frac{16 \pi M}{(d-2) \Omega_{d-2}},
\end{align*}
where $\Omega_{d-2}$ is the volume of the unit sphere in $(d-2)$ dimensions. 
We will restrict our attention to black holes which are \emph{subcritical} - that is, where the mass is positive.

 The event horizon for the black hole is a surface at which the Killing vector field $\pr_t$ vanishes. It is located at a radius
\begin{align*}
r=r_H &:= \mu^{\frac{1}{d-3}}.
\end{align*}
In this paper, we will only be interested in the domain of outer communication (DOC) for the Schwarzschild-Tangherlini black hole, which is the spacetime patch where $r > r_H$. Note that in the DOC for a subcritical Schwarzschild-Tangherlini black hole there are no curvature singularities. The (smooth) Lorentzian manifold consisting of the DOC equipped with the Schwarzschild-Tangherlini metric will be denoted by $(\mathcal{M} , g)$.

Unstable, trapped null geodesics can exist in these spacetimes at constant radius \cite{DFR11}
\begin{align}
\label{eq:SchwTangphotonradius}
r = r_P &:= \left( \frac{\mu(d-1)}{2} \right)^{\frac{1}{d-3}}.
\end{align}
The sphere of radius $r = r_P$ at a fixed time is known as the photon sphere.

\section{The photon region in the Schwarzschild-Tangherlini black holes}
\label{sec:Results}
The photon sphere provides examples of trapped photons in the Schwarzschild-Tangherlini black hole spacetimes - that is, photons which do not fall into the black hole or escape to spatial infinity. In order to prove that these are the \emph{only} trapped photons, we follow \cite{CJ19} in precisely defining the notion of a trapped photon.
\begin{definition}
	A photon in the DOC of a stationary spacetime is called \emph{trapped} if its orbit in the quotient of the DOC under the action of the stationary Killing vector field $\pr_t$ is contained in a compact set.
\end{definition}
For the Schwarzschild-Tangherlini spacetimes, a photon is trapped if and only if the range of its radial coordinate is contained in a relatively compact subset of $(r_H,\infty)$. We now proceed to prove our first result. 

\begin{proposition}
\label{prop:schwtangsphere}
The spherical photon orbits are the only trapped photons in the DOC of a subcritical Schwarzschild-Tangherlini black hole.
\end{proposition}
\begin{proof}
Let us assume, for the sake of contradiction, that there is a trapped photon $\gamma = (t,r,\boldsymbol{\theta})$ in the DOC which has non-constant coordinate radius. Since it is a trapped photon, the range of the radial coordinate must be a relatively compact set in $(r_H , \infty)$. The radius is non-constant but contained within a relatively compact set in $(r_H, \infty)$, so there must be some $r_1$ and $r_2$ such that $r_H < r_1 < r_2 < \infty$ and for which $r \in [r_1, r_2]$. The geodesic equations for this background are computed in \cite{SG17}, and the radial geodesic equation is given by
\begin{align*}
r^4 \dot{r}^2 &=\mathcal{R}(r) := E^2 r^4 - r^2 \left( 1- \frac{\mu}{r^{d-3}} \right)(\mathcal{K} + \Phi^2).
\end{align*}
In order for the radial geodesic equation to have real solutions, we require $\mathcal{R}$ to be positive in $[r_1,r_2]$. At the endpoints $r_i$ of this interval we must have either $\dot{r} = 0$ (corresponding to a radial turning point for the photon), or else we must have $\lim_{s\to+\infty} r(s) = r_i$ or $\lim_{s\to-\infty} r(s) = r_i$, (corresponding to the photon asymptotically approaching this radial value). Either way, this corresponds to having a zero of $\dot{r}$, and therefore $\mathcal{R}$, at $r_i$. The zeros of the function $\mathcal{R}$ are the same as the zeros of the effective potential, $V_{eff}$, defined in \cite{SG17}:
\begin{align*}
V_{eff} &:= -\frac{\mathcal{R}(r)}{r^4} \\
&= \frac{1}{r^2} \left( 1-\frac{\mu}{r^{d-3}} \right)(\mathcal{K} + \Phi^2) - E^2.
\end{align*}
A direct computation reveals that there is precisely one critical point of $V_{eff}$, located at $r = r_P$. We also note that $V(r_H) \leq 0$, and $\lim_{r\to\infty}V_{eff}(r) \leq 0$, since $E$ is a real quantity. It follows that there are two real distinct roots of $V_{eff}$ in $[r_1,r_2]$ if and only if $V_{eff}(r_P)>0$. This gives us our required contradiction, however, since this implies that $\mathcal{R}(r_P) < 0$.
\end{proof}
The region in spacetime where one can find trapped photons is generally referred to as the \emph{photon region}. We have just proven for Schwarzschild-Tangherlini black holes, the spatial slice of the photon region is just the photon sphere. 

In order to prove our second result, which characterises the topology of the photon region as a submanifold of the tangent bundle \`{a} la \cite{CJ19}, we must first discuss the identification of geodesics with points in the tangent bundle. Recall that a point in the tangent bundle $T\mathcal{M}$ can be specified by a pair $(x,v)$, where $x$ is a point on the manifold $\mathcal{M}$, and $v$ is a tangent vector to $\mathcal{M}$ at $x$. Then given a geodesic $\gamma$, we identify the corresponding point in $T\mathcal{M}$ by
\begin{align*}
\gamma \mapsto \Big( \gamma(0) , \dot{\gamma}(0) \Big),
\end{align*}
where we distinguish between geodesics of difference affine parametrisations. This identification is a bijection since geodesics are uniquely specified by $\Big( \gamma(0) , \dot{\gamma}(0) \Big)$. For the next result, we recall that the unit tangent bundle $T^1 \mathcal{N}$ of a Riemannian manifold $(\mathcal{N},g)$ is a fiber bundle over $\mathcal{N}$ whose fiber is the unit sphere in the tangent bundle. That is,
\begin{align*}
T^1 \mathcal{N} &:= \coprod_{x \in T \mathcal{N}} \{ v \in T_x \mathcal{N} : g_{x} (v,v) = 1 \} 
\end{align*}
We shall now prove that for Schwarzschild-Tangherlini black holes, the photon region in the tangent bundle has the topology suggested in \cite{CJ19}. 
\begin{proposition}
	\label{prop:schwtangtopology}
	Let $(\mathcal{M},g)$ be a Schwarzschild-Tangherlini black hole of dimension $d = n+1$. The photon region in the tangent bundle, $\mathcal{P} \subset T\mathcal{M}$, has topology $T^1 S^{n-1} \times \RR^2$.
\end{proposition}
\begin{proof}
	Since the photon region is invariant under time translation generated by $\pr_t$ and under rescaling the energy 
	\begin{align*}
	(E, \dot{t} , \dot{r}, \dot{\theta}_1, \dot{\theta}_2, \dots,  \dot{\theta}_{d-2}) \mapsto \lambda 	(E, \dot{t} , \dot{r}, \dot{\theta}_1, \dot{\theta}_2, \dots,  \dot{\theta}_{d-2}),
	\end{align*}
	we can restrict our attention to the $2n$-dimensional slice 
	\begin{align*}
	\mathcal{P}_0 &= \mathcal{P} \cap \left\{ t=0, \, p_{t} = -E = -\sqrt{\frac{d-1}{d-3} } \right\}
	\end{align*}
	in the tangent bundle.\footnote{The reason for this seemingly strange scaling of $E$ will become clear below.} Trapped null geodesics for Schwarzschild-Tangherlini black holes all have constant coordinate radius $r = r_P$, and so the photon region is a bundle over an $(n-1)$-sphere of radius $r_P$, with as-yet unknown fibers. At a point on the photon sphere, an arbitrary trapped geodesic has the form	
	\begin{align*}
	\gamma &= (t ,r ,\boldsymbol{\theta}, \dot{t} , \dot{r} ,  \dot{\boldsymbol{\theta}} ) = \left( 0 , r_P , \boldsymbol{\theta} ,  -E, 0 , \dot{\boldsymbol{\theta}} \right),
	\end{align*}
	where $\boldsymbol{\theta}$ are angular coordinates on the photon sphere. The length-squared of such a geodesic is
	\begin{align*}
	g(\dot{\gamma} , \dot{\gamma} ) &= -\left( \frac{d-3}{d-1} \right) E^2 + g_{r_P}(\dot{\boldsymbol{\theta}} , \dot{\boldsymbol{\theta}} ),
	\end{align*}
	where $g_{r_P}$ is the (Riemannian) round metric on the photon sphere induced by the Schwarzschild-Tangherlini metric. It follows that such a geodesic is null provided 
	\begin{align*}
	g_{r_P}(\dot{\boldsymbol{\theta}} , \dot{\boldsymbol{\theta}} ) &= \left( \frac{d-3}{d-1} \right) E^2 = 1
	\end{align*}
	That is, $\mathcal{P}_0$ is the unit tangent bundle over $S^{n-1}$, and the photon region has the topology $$T^1 S^{n-1} \times \RR^2.$$
\end{proof}
Note that $T^1 S^{n-1}$ is always an $S^{n-2}$-bundle over $S^{n-1}$.
\begin{example}
	The four-dimensional case was covered in detail in \cite{CJ19}. The unit tangent bundle $T^1 S^2$ is a non-trivial circle bundle over $S^2$ with Euler class $2$, and is therefore homeomorphic to the Lens space $L(2;1) \simeq SO(3)$. It follows that the topology of $\mathcal{P}$ is $SO(3) \times \RR^2$.
\end{example}
\begin{example}
	In five-dimensions, the unit tangent bundle $T^1 S^3$ is an $S^2$--bundle over $S^3$. An old result of Steenrod says that such a bundle must be trivial \cite{S44}, so we find that the topology of $\mathcal{P}$ must be $S^2 \times S^3 \times \RR^2$. 
\end{example}
\begin{example}
	In six dimensions, the unit tangent bundle $T^1 S^4$ is an $S^3$-bundle over $S^4$. This bundle is sometimes referred to as $M_{-1,2}$, and is (orientation reversing) diffeomorphic to ``twice'' the Hopf fibration \cite{CE00}. That is, the topology of $\mathcal{P}$ is $S^7/\ZZ_2 \times \RR^2$.
\end{example}


\section{Summary and Discussion}
\label{sec:summary}

In this paper we have proven two key results about the set of trapped photons outside the event horizon of Schwarzschild-Tangherlini black holes in arbitrary dimension $d \geq 4$. The first result states that the only trapped photons are the photons of constant coordinate radius - that is, the only trapped photons must lie on the photon sphere at radius $r = r_P$. The second result is a characterisation of the topology of the photon region as a submanifold of the tangent bundle of spacetime. It states that the topology of the photon sphere in the tangent bundle is $T^1 S^{n-1} \times \RR^2$, as suggested in \cite{CJ19}.\\

The first result says that the spatial slice of the photon region in Schwarzschild-Tangherlini black holes is just the photon sphere. The photon region is therefore a smooth submanifold of the spacetime, equipped with a smooth metric induced from the ambient spacetime metric, and we understand both its geometry and topology very well. Given this, one may wonder why the second result is useful or necessary. If we understand the geometry as a submanifold of spacetime, why do we need to complicate matters by moving to the tangent bundle? The answer to this question lies in the generalisation to rotating black holes. As shown in \cite{CJ19}, in four dimensions the photon region in a Kerr black hole is not a submanifold of spacetime, but \emph{is} a submanifold of the tangent bundle of spacetime. In particular, the photon region of the Kerr black hole has the same topology of the Schwarzschild photon region \cite{CJ19}. Like the four-dimensional case, the photon region for higher-dimensional rotating black holes is not a submanifold of the spacetime, but one may hope that it is a submanifold of the tangent bundle with the same topology as the non-rotating case. Work on extending the results of \cite{CJ19} to rotating black holes in five dimensions is currently in progress \cite{B20}.


 \section*{Acknowledgements}

Research of M.B. was supported by the GA\v{C}R Grant EXPRO 19-28628X.

\end{document}